\documentclass[letterpaper,twocolumn,10pt]{article}
\usepackage{usenix2019_v3}

\usepackage{amsmath,amssymb,amsfonts,amsthm}
\usepackage{graphicx}
\usepackage{booktabs}
\usepackage{tabularx}
\usepackage{array}
\usepackage{float}
\usepackage{listings}
\usepackage{algorithm}
\usepackage{algpseudocode}
\usepackage{subcaption}
\usepackage{xspace}
\usepackage{tikz}
\usetikzlibrary{arrows.meta, positioning, calc, shapes.geometric, fit, shadows, matrix}

\microtypesetup{spacing=false}

\AtBeginDocument{\DeclareMathAlphabet{\mathcal}{OMS}{cmsy}{m}{n}}

\newcommand{\PCI}{\textnormal{\textsc{PCI}}\xspace}
\newcommand{\CK}{\textnormal{\textsc{CK}}\xspace}

\newtheorem{proposition}{Proposition}

\newcolumntype{L}[1]{>{\raggedright\arraybackslash}p{#1}}
\newcolumntype{C}[1]{>{\centering\arraybackslash}p{#1}}
\newcolumntype{R}[1]{>{\raggedleft\arraybackslash}p{#1}}

\definecolor{navy}{RGB}{20,50,110}
\definecolor{crimson}{RGB}{180,30,30}
\definecolor{slate}{RGB}{112,128,144}
\definecolor{emerald}{RGB}{34,139,34}
\definecolor{darkpurple}{RGB}{102,51,153}
\definecolor{sublayer}{RGB}{235,240,250}
\definecolor{kernellayer}{RGB}{250,235,235}

\hypersetup{
  pdftitle={Beyond Memory: A Transactional Continuity Kernel for Long-Lived AI Agents},
  pdfauthor={Jun He; Deying Yu},
  pdfsubject={Persistent Agent State, State Transition Protocols, Optimistic Concurrency, Atomic Commit},
  pdfkeywords={persistent agent state, continuity kernel, state transition protocol, exact-head concurrency, writer fencing, receipts}
}

\begin{document}

\title{\bf Beyond Memory: A Transactional Continuity Kernel for Long-Lived AI Agents}

\author{
  {\rm Jun He}\\
  OpenKedge.io
  \and
  {\rm Deying Yu}\\
  OpenKedge.io
}

\maketitle

\begin{abstract}
Persistent AI agents accumulate versioned state across long horizons, but storage retention alone does not identify authoritative state. Without an explicit control plane, unmediated updates by models, tools, and background workers risk stale overwrites, un-audited exposures, and self-authorizing privilege escalation. We argue that agent state governance is an infrastructural activation problem, defining \emph{continuity} as an unbroken, authorized lineage of accepted branch heads. We present the \emph{Continuity Kernel} (\CK), an activation contract that decouples off-commit candidate evaluation from atomic state activation. Untrusted components propose typed changes against an exact predecessor head or typed absence. A short activation transaction revalidates ownership, pre-state authority, freshness, and effect uniqueness, recording one stable disposition (\textsf{Commit}, \textsf{Reject}, \textsf{Quarantine}, or \textsf{Defer}). Only \textsf{Commit} atomically advances the branch head and installs the complete accepted unit (state, authority, lineage, effects, outcome, and receipt). A bounded executable model verifies the protocol across 2,808,230 reachable states and 5,526,474 state-changing transitions with zero invariant violations.
\end{abstract}

\begin{figure*}[t]
\centering
\definecolor{navy}{RGB}{0, 32, 96}
\definecolor{sublayer}{RGB}{235, 243, 250}
\definecolor{crimson}{RGB}{192, 0, 0}
\definecolor{kernellayer}{RGB}{253, 235, 235}
\definecolor{darkpurple}{RGB}{84, 34, 117}

\begin{tikzpicture}[
  node distance=1.1cm and 0.8cm,
  base/.style={
    rounded corners=3pt,
    minimum height=1cm,
    text width=2.8cm,
    align=center,
    thick,
    font=\small\sffamily
  },
  box/.style={base, draw=navy, fill=sublayer},
  kernel/.style={base, draw=crimson, fill=kernellayer},
  store/.style={base, draw=darkpurple, fill=white},
  arr/.style={-{Stealth[length=2.5mm, width=1.5mm]}, thick, draw=navy!80!black},
  kernel-arr/.style={-{Stealth[length=2.5mm, width=1.5mm]}, thick, draw=crimson},
  feedback-arr/.style={-{Stealth[length=2.5mm, width=1.5mm]}, thick, dashed, draw=darkpurple}
]

\node[box] (prop) {Models, tools, operators\\[0.5ex]untrusted proposers};
\node[box, right=of prop] (prep) {Authenticate proposal\\[0.5ex]and identifier};
\node[box, right=of prep] (eval) {Acquire evidence\\[0.5ex]and seal candidate};
\node[kernel, right=of eval] (activate) {Revalidate and\\[0.5ex]decide atomically};

\node[store, above right=-0.1cm and 1cm of activate] (state) {Authoritative state\\[0.5ex]and branch head};
\node[store, below right=-0.1cm and 1cm of activate] (attempt) {Outcome, receipt,\\[0.5ex]and effect index};

\draw[arr] (prop) -- (prep);
\draw[arr] (prep) -- (eval);
\draw[arr] (eval) -- (activate);

\draw[kernel-arr] (activate.east) to[out=0, in=180] (state.west);
\draw[kernel-arr] (activate.east) to[out=0, in=180] (attempt.west);

\draw[feedback-arr] (state.north) to[out=150, in=60] (activate.north);

\end{tikzpicture}
\caption{Probabilistic proposal generation and remote evaluation precede the
short activation transaction. Only activation advances an authoritative head.}
\label{fig:minimal-architecture}
\end{figure*}
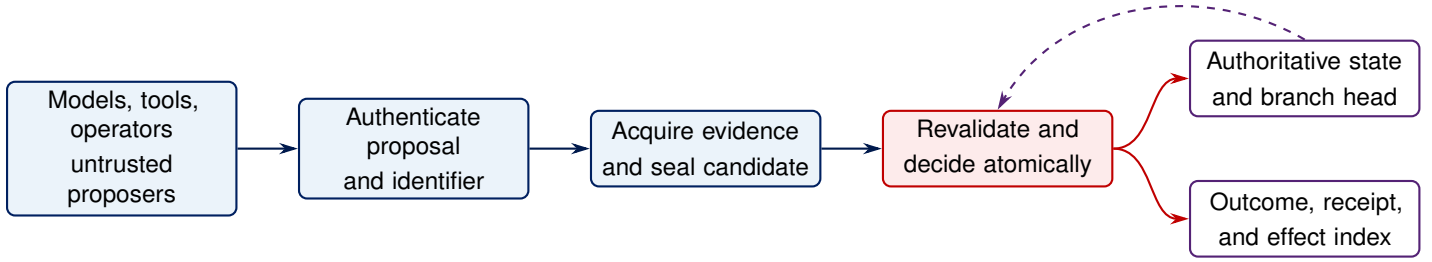

\section{Introduction}
\label{sec:introduction}

Long-lived AI agents increasingly retain memories, profiles, plans, tool state,
and policies beyond a single context window~\cite{park2023generative,packer2023memgpt}.
However, high retrieval quality does not determine which state version is authoritative
when models, tools, compactors, and background recovery workers submit concurrent or conflicting updates.
Without a strict mutation boundary, a storage layer may preserve every version yet accept a stale retry,
expose uncommitted state without audit trails, or allow a proposal to authorize itself by injecting the required privileges into its own proposed state.

We define \emph{continuity} as an unbroken, authorized lineage of
accepted heads on one branch. This is infrastructural continuity---a guarantee about state activation and lineal provenance---not a philosophical claim about agent consciousness or behavioral identity.
Continuity fundamentally differs from retention: rejected and quarantined objects may remain stored for diagnostic or policy reasons, but they remain unreachable from the authoritative branch head. Cryptographic commitments can enforce this reachability boundary; they cannot make an inferred memory factually true or a policy normatively legitimate.

The core thesis of this paper is that agent state governance is an infrastructural activation contract problem. To operationalize this principle, we define the \emph{Continuity Kernel} (\CK) as an activation contract that separates off-commit candidate \emph{proposal} from atomic \emph{activation}. Models, tools, and operators act as untrusted proposers. They may perform probabilistic reasoning, remote network fetches, and candidate state construction, but only the deterministic \CK control plane may advance an authoritative branch head.
Each proposal targets an exact predecessor head or typed target absence.
Slow evaluation and evidence acquisition precede activation; a short transaction then revalidates all mutable facts governing admission and records exactly one durable disposition. Figure~\ref{fig:minimal-architecture} illustrates this activation boundary.

\CK introduces no new low-level transaction engine. Instead, standard substrates---such as relational database transactions, conditional object manifests, or consensus state-machine logs---execute its atomic step. The primary contribution is the agent-state activation contract enforced within that step: it atomically binds proposal identity, exact predecessor state, pre-state authority, acquired evidence, lifecycle status, effect manifests, lineage, and outcome receipts. Specifically, this paper makes three contributions:
\begin{enumerate}
  \item \textbf{System Model and Contract:} A formal model separating stored retention from authoritative reachability, defining the boundary between untrusted proposers and trusted activation (Section~\ref{sec:system-model});
  \item \textbf{Activation Protocol:} An owner-bound, exact-head activation protocol with pre-state authorization, four terminal dispositions (\textsf{Commit}, \textsf{Reject}, \textsf{Quarantine}, \textsf{Defer}), at-most-once effect execution, and explicit receipt evidence levels (Section~\ref{sec:transition-protocol}); and
  \item \textbf{Lifecycle Rules and Model Verification:} Typed lifecycle semantics for branch creation, writer handoff, schema migration, and forward restoration, verified through bounded state-space exploration (Sections~\ref{sec:branch-isolation-fencing}--\ref{sec:evaluation}).
\end{enumerate}

Sections~\ref{sec:system-model}--\ref{sec:evaluation} develop the core architecture and evaluation; Appendices~\ref{app:receipt-semantics}--\ref{app:resource-artifact} detail normative specifications, lineage proofs, and model parameters.
\section{System Model and Contract}
\label{sec:system-model}

\subsection{Authority Is Reachability}

As illustrated in Figure~\ref{fig:minimal-architecture}, \CK establishes a strict boundary between candidate object preparation and authoritative state activation. Agent state is partitioned into isolated branches identified by a subject/system identifier $sid$ and a branch identifier $bid$, denoted $k=(sid,bid)$. The branch-indexed state $X[k]$ holds the complete collection of typed agent components for branch $k$, including memory records, user profiles, tool configurations, and policy-bearing objects.

As summarized in Table~\ref{tab:state-namespaces}, the storage layer maintains multiple functional namespaces. Candidate proposals, execution attempt logs, quarantined candidates, and read-only runtime projections coexist in storage alongside accepted state. However, storage presence alone does not confer authority: an object is authoritative if and only if it is reachable from the current branch head committed by \CK.

\begin{table}[t]
\centering
\caption{Storage does not by itself confer authority.}
\label{tab:state-namespaces}
\small
\begin{tabularx}{\columnwidth}{L{1.7cm} X C{1.25cm}}
\toprule
\textbf{Namespace} & \textbf{Purpose} & \textbf{In head?} \\
\midrule
Accepted & Current typed state & Yes \\
Candidate & Prepared successor & No \\
Attempt & Outcome and diagnostics & No \\
Quarantine & Isolated candidate & No \\
Projection & Read-only runtime view & No \\
\bottomrule
\end{tabularx}
\end{table}

For branch $k=(sid,bid)$, the branch directory maintains the complete-head type $h$ and branch-row type $B[k]$:
\begin{equation}
\begin{aligned}
h &= \langle sid, bid, seq, root, \boldsymbol v, parentRef, lineageRef \rangle, \\
B[k] &= \langle h, status, writer, epoch, dirSeq, handoff, createdFrom \rangle.
\end{aligned}
\label{eq:compact-head}
\end{equation}
Here $root$ commits the typed component state; $\boldsymbol v=(schemaV, policyV, evalV, authV)$ names the pinned versions; $parentRef$ is the parent head reference; $lineageRef$ binds the causal lineage; $dirSeq$ is directory sequence; $handoff$ is the open revision pointer; and $createdFrom$ records genesis provenance. Heads are equal only when every canonically encoded field agrees.

The abstract kernel configuration is
\begin{equation}
C = \langle X, B, \Gamma, \mathcal S, t \rangle.
\label{eq:compact-configuration}
\end{equation}
$\Gamma$ is the pre-state authority context; $\mathcal S$ represents persistent metadata stores (outcomes, receipts, lineage, effects, and quarantine); and $t$ is trusted transaction time. An activation serializes the branch row and every admission, allocation, and effect row that its proposal names.

\subsection{Proposals and Typed Transitions}

Models, tools, memory services, and operators are untrusted proposers. Evaluators and approvers issue signed evidence but cannot advance a head. A trusted gateway authenticates proposal ownership; a preparation service acquires the required evidence and derives a candidate; the activation engine alone changes authority.

A signed proposal is summarized by
\begin{equation}
\tau = \langle pid, target, expected, ops, a, req \rangle.
\label{eq:compact-proposal}
\end{equation}
Here $pid$ is the unique proposal identifier; $target$ is the target branch key $k=(sid,bid)$; $expected$ is the expected predecessor head; $ops$ is an ordered sequence of state operations; $a$ is a sequence of authority changes; and $req$ declares the evidence required for admission.

Candidate derivation occurs off-commit, producing candidate state $X' = F(X, ops, W)$ using evidence $W$. Crucially, authorization evaluates against the pre-state context $A$:
\begin{equation}
X' = F(X, ops, W), \qquad A = \begin{cases} \Gamma, & \text{existing branch}, \\ \xi_k, & \text{branch creation}. \end{cases}
\label{eq:compact-transition}
\end{equation}
$\operatorname{Authorize}_A(\tau,W)$ runs before computing the authority update $\Gamma'=F_A(A,a)$; a proposal cannot authorize itself. The candidate seal binds the proposal, evidence, interpreter, and component roots. Activation recomputes these bindings against the serialized branch state.

Schema evolution is an explicit Migration transition with a pinned, deterministic migrator. Deletion creates a typed tombstone; physical erasure remains subject to retention policy. Hashes establish cryptographic integrity, not confidentiality, semantic truth, or policy legitimacy.

\subsection{Threat Model and Conditional Contract}

Proposers may be buggy or adversarial: they may replay requests, reuse identifiers, reorder operations, bind a stale head, omit evidence, forge scope, or request self-authorizing changes. Infrastructure may crash, lose replies, duplicate messages, expose stale replicas, or partition. Evaluators may be wrong; their outputs are policy inputs rather than semantic oracles.

The safety claims depend on nine explicit assumptions (A1--A9, stated individually in Appendix~\ref{app:resource-artifact}):
(1)~\textbf{Boundary \& Cryptography (A1--A3):} all authoritative writes cross an authenticated kernel boundary, typed cryptography is sound, and signing keys are protected;
(2)~\textbf{Atomic Serialization (A4):} the storage substrate atomically serializes the complete activation key set;
(3)~\textbf{Complete Context \& Time (A5--A6):} policy, authority, revocation, lifecycle, dependency, and trusted-time values are available for commit-time validation;
(4)~\textbf{Lifecycle Order \& Non-Recycled Scopes (A7--A8):} lifecycle operations share one order and proposal/effect scopes and writer epochs are not recycled; and
(5)~\textbf{Verification Objects (A9):} stronger receipt claims are made only when their required verification objects are available.

Under those assumptions, the contract has four consequences:
\begin{enumerate}
  \item \textbf{Exact succession.} One complete predecessor has at most one accepted successor; one serialized absence has at most one genesis. Each is installed with state, lineage, effects, outcome, and receipt.
  \item \textbf{Pre-state admission.} Authorization and freshness are checked against the serialized predecessor or creation context, never the proposed context.
  \item \textbf{Stable execution identity.} A proposal identifier has one terminal disposition and an effect identifier has at most one accepted binding, including after safe reclamation.
  \item \textbf{Lifecycle isolation.} Branch creation, fencing, handoff, migration, and restoration preserve branch scope and current authority.
\end{enumerate}
These are safety properties, not availability guarantees. The kernel may be unavailable or may return Defer when evidence cannot be obtained. It governs only state behind its mutation boundary; prompts, caches, model weights, tool-local data, prior disclosures, and remote side effects remain outside unless separately mediated.

\section{Activation Protocol}
\label{sec:transition-protocol}

The protocol separates slow, fallible preparation from one short serialized
decision. Preparation authenticates the proposer, acquires exactly the evidence
declared by the proposal, runs the pinned transition function, and seals the
candidate. Activation admits that sealed package only if the relevant state is
still current. Appendix~\ref{app:receipt-semantics} fixes the complete stage
order and receipt variants; this section states the activation rule.

\subsection{One activation predicate}

Let $P=\langle\tau,W,X',M_E,s\rangle$ be a prepared package containing the
signed proposal, exact evidence set, candidate state, ordered effect manifest,
and trusted candidate seal. Inside the activation transaction, the kernel
evaluates the ordered vector
\begin{equation}
\mathcal G_{kind}(C,P)=
\langle G^{kind}_a(C,P)\rangle_{a\in\mathsf{ActCheck}}.
\label{eq:activation-guard}
\end{equation}
The kernel evaluates this vector strictly in the Appendix~A order and accepts
only an all-\textsf{Pass} vector. Table~\ref{tab:guard-groups} gives one
predicate per stage; no aggregate guard can move a failure across stages.

\begin{table*}[t]
\centering
\caption{One commit-time predicate per activation stage; read down the left
pair, then down the right pair.}
\label{tab:guard-groups}
\scriptsize
\begin{tabularx}{\textwidth}{L{2.4cm} X L{2.4cm} X}
\toprule
\textbf{Stage} & \textbf{Predicate} & \textbf{Stage} & \textbf{Predicate} \\
\midrule
\textsf{TargetLookup} & Existing kind: $B[k]$ exists. BranchCreate: $B[k]$ is
absent and serialized $\Xi[k]$ exists. &
\textsf{CandidateBinding} & Seal binds proposal, evidence, interpreter,
candidate root, and applicable creation-context or open-revision digest. \\
\textsf{Allocation} & Owner allocation is current; $pid/eid$ names are live,
unused, and unretired. & \textsf{EffectBinding} & Operations and manifest are
bijective; effect scopes are live and unused. \\
\textsf{ExpectedHead} & Existing kind: $\mathsf{Present}(d_h)$ is exact;
BranchCreate: expectation is $\mathsf{Absent}$. &
\textsf{FreshnessInitial} & Revalidate time and the kind-indexed snapshot:
head/epoch/$dirSeq$/authority, or absence/$d_\xi$/genesis inputs. \\
\textsf{Lifecycle} & Status, writer, epoch, and kind-specific lifecycle row
are admissible. & \textsf{Authorization} & $A_\chi$---predecessor $\Gamma$ or
serialized creation context $\xi_k$---authorizes the operations. \\
\textsf{HandoffRevision} & Applicable target, open SourceFenced revision, reserved epoch,
and $dirSeq$ are exact. & \textsf{Decision} & Policy records Proceed, Reject,
or Quarantine independently of authorization. \\
\textsf{AcquisitionAuth} & Acquisition result set and signer are authentic. &
\textsf{AuthorityTransition} & A declared authority change is valid under the
authorized kind context $A_\chi$. \\
\textsf{RequirementCoverage} & Signed requirements and results form an exact
ordered cover. & \textsf{FreshnessFinal} & Repeat the complete kind-indexed
snapshot immediately before prospective construction. \\
\textsf{VersionBinding} & Policy, evaluator, authority, revocation,
issuer/allocation, dependency, and processing versions are current. &
\textsf{WellFormedness} & Receipt-independent prospective unit $P_\chi$ is
complete, deterministic, and well typed. \\
\textsf{WitnessBinding} & Every required witness or authenticated Missing is
bound exactly. & & \\
\bottomrule
\end{tabularx}
\end{table*}

The trusted acquisition service returns an exact one-to-one mapping for the declared requirements: no declared requirement is omitted, and no unrequested evidence is included. Any unavailable evidence is returned as an explicit, authenticated \textsf{Missing} status, preventing callers from manufacturing \textsf{Defer} by suppressing inconvenient facts or injecting extraneous context. Mutable requirements specify explicit serialization keys and versions, which \textsf{VersionBinding} revalidates at commit time. Freshness is validated prior to policy evaluation and re-checked immediately before prospective state construction. Ordinary transitions revalidate the exact branch head snapshot, while \textsf{BranchCreate} revalidates target absence, creation context $\xi_k$, allocator versions, and genesis inputs.

The candidate seal certifies that the candidate state was derived by the pinned deterministic interpreter from the signed inputs and acquired evidence. Authorization is a separate commit-time predicate evaluated over the pre-state authority context $A_\chi$. The subsequent \textsf{Decision} stage records the policy result, and an authority transition $\Gamma'$ is computed only after both authorization and decision stages succeed.

\subsection{Four stable dispositions}

As illustrated in the proposal transition lifecycle (Figure~\ref{fig:activation-lifecycle-wide}), for an authenticated, owner-bound identifier, the kernel records exactly one of four dispositions:
\begin{equation}
q\in\{\mathsf{Commit},\mathsf{Reject}(r),
\mathsf{Quarantine}(r),\mathsf{Defer}(r)\}.
\label{eq:dispositions}
\end{equation}

\begin{figure*}[t] 
\centering
\definecolor{navy}{RGB}{0, 32, 96}
\definecolor{sublayer}{RGB}{235, 243, 250}

\begin{tikzpicture}[
  base/.style={
    rounded corners=3pt, 
    minimum height=0.7cm, 
    text width=1.8cm, 
    align=center, 
    thick, 
    font=\scriptsize\sffamily
  },
  state/.style={base, draw=navy, fill=sublayer},
  term/.style={base, draw=navy!70!black, fill=white},
  commit/.style={base, draw=navy!90!black, fill=sublayer, font=\bfseries\scriptsize\sffamily},
  arr/.style={-{Stealth[length=2.2mm, width=1.3mm]}, thick, draw=navy!80!black},
  edgelabel/.style={font=\tiny\sffamily, fill=white, inner sep=2pt}
]

\node[state] (prep) {Preparation\\[0.3ex]\tiny Off-commit};
\node[term, below=0.8cm of prep] (defer) {Defer\\[0.3ex]\tiny Missing data};

\node[state, right=2.2cm of prep] (act) {Activation\\[0.3ex]\tiny Atomic check};

\node[term, right=2.2cm of act] (reject) {Reject\\[0.3ex]\tiny Terminal fail};
\node[commit, above=0.4cm of reject] (commit) {Commit\\[0.3ex]\tiny Advances head};
\node[term, below=0.4cm of reject] (quar) {Quarantine\\[0.3ex]\tiny Retain unlinked};


\draw[arr] (prep) -- node[edgelabel] {Sealed} (act);
\draw[arr] (prep) -- node[edgelabel] {Missing} (defer);

\draw[arr] (defer.west) to[out=180, in=180, looseness=1.5] node[edgelabel, left=2pt] {New $pid$} (prep.west);

\draw[arr] (act.east) to[out=0, in=180] node[edgelabel, pos=0.45] {Pass} (commit.west);
\draw[arr] (act.east) -- node[edgelabel, pos=0.45] {Fail} (reject.west);
\draw[arr] (act.east) to[out=0, in=180] node[edgelabel, pos=0.45] {Policy} (quar.west);

\end{tikzpicture}
\caption{Proposal transition lifecycle and four terminal dispositions, expanding chronologically left-to-right.}
\label{fig:activation-lifecycle-wide}
\end{figure*}
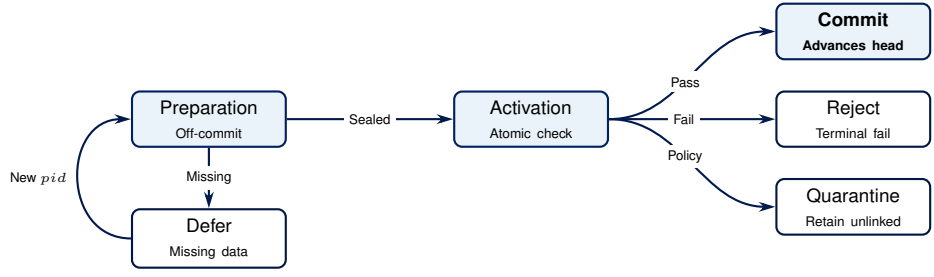

\begin{table}[t]
\centering
\caption{Only Commit changes authoritative state.}
\label{tab:dispositions}
\small
\begin{tabularx}{\columnwidth}{L{1.45cm} X}
\toprule
\textbf{Result} & \textbf{Meaning} \\
\midrule
Commit & Install the exact successor and all acceptance metadata atomically. \\
Reject & A permanent validation or policy failure; retain no successor. \\
Quarantine & Retain sealed material outside the authoritative head. \\
Defer & Trusted acquisition reports required evidence unavailable. \\
\bottomrule
\end{tabularx}
\end{table}

Every disposition is terminal for its proposal identifier. Resubmission after
Defer or Quarantine uses a new identifier linked to the old attempt; a lost
response is resolved by looking up the original identifier. Malformed outer
framing, failed outer authentication, or an invalid allocation capability are
protocol errors rather than dispositions and must not consume the identifier.
After ownership succeeds, a fixed stage order chooses the result: any
permanent failure that can already be evaluated precedes a trusted Missing;
later predicates remain unevaluated. This prevents a missing witness from
masking a known invalid proposal.

\begin{algorithm}[t]
\caption{Prepare and activate a proposal (abstract)}
\label{alg:activate}
\small
\begin{algorithmic}[1]
\Require signed proposal $\tau$ and authenticated allocation
\State authenticate the principal and ownership of $pid$
\State return a prior outcome or identifier conflict, if present
\State run \textsf{PrepStage} to Ready; acquire and derive only at named stages
\State on terminal preparation Fail/Missing, durably \Return $q(r_P)$
\State begin a transaction over the branch and every named mutable key
\State return the prior outcome if another retry installed it
\State read every row named by $\mathsf{ActCheck}$
\For{$a$ in the normative $\mathsf{ActCheck}$ order}
  \If{$a=\mathsf{WellFormedness}$}
    \State $\widehat P_\chi\gets\mathsf{BuildProspective}_\chi(P,C)$
  \EndIf
  \State evaluate only $G^{kind}_a$ and append its stage result
  \If{$a=\mathsf{Decision}$ and result is Quarantine}
    \State durably \Return $\mathsf{Quarantine}(r_Q)$
  \ElsIf{result is Fail or Missing}
    \State durably \Return $q(r_A)$
  \EndIf
\EndFor
\State $r_C\gets rb_C(\widehat P_\chi)$; compute $d_{r_C}$
\State $\mathcal U_{accept}(\chi)\gets\mathsf{Finalize}_\chi(\widehat P_\chi,d_{r_C})$
\State atomically install the complete unit; \Return $\mathsf{Commit}(r_C)$
\end{algorithmic}
\end{algorithm}

Algorithm~\ref{alg:activate} abstracts concrete wire encodings and storage driver APIs.
After all validation stages pass, the kernel constructs the receipt-independent prospective unit $\widehat P$:
\begin{equation}
\widehat P = \mathsf{BuildProspective}(X', h, B[k]).
\label{eq:prospective-unit}
\end{equation}
Construction follows a strictly cycle-free order:
\begin{equation}
\widehat P \longrightarrow \textsf{WellFormedness}(\widehat P) \longrightarrow r_C \longrightarrow h' \longrightarrow \mathcal U_{\text{accept}}.
\label{eq:prospective-order}
\end{equation}
First, the Commit receipt $r_C$ is derived directly over $\widehat P$, binding candidate state $X'$, evidence, and pinned versions without circular dependency on the receipt itself. Second, the complete head $h'$ is finalized by binding $r_C$ and lineage. Finally, the kernel installs the complete accepted unit $\mathcal U_{\text{accept}}$:
\begin{equation}
\mathcal U_{\text{accept}} = \{ X', \Gamma', h', B'[k], r_C, O[pid] \}.
\label{eq:atomic-accepted-unit}
\end{equation}
The set $\mathcal U_{\text{accept}}$ denotes one all-or-nothing atomic transaction, installing candidate state $X'$, updated authority context $\Gamma'$, finalized branch row $B'[k]$ pointing to head $h'$, execution outcome $O[pid]$, and receipt record $r_C$. Any non-Commit disposition leaves $X$, $B$, and $\Gamma$ unchanged.

Proposal and effect records need not remain online forever. Reclamation first advances a durable, monotonically increasing watermark for a non-recycled allocation scope and only then removes covered rows. An identifier below that watermark returns RetiredIdentifier rather than re-entering execution. This preserves stable identity while allowing bounded online metadata.

\subsection{Receipts state what they prove}
\label{sec:receipt-verification}

A receipt binds the proposal, disposition, reached protocol stage, decisive
reason, and the evidence available at that stage. Commit receipts additionally
bind the applicable typed parent and context, successor core, lineage,
authority versions, and effect manifest. Preparation failures need not pretend
that a canonical proposal or candidate existed.

Receipt verification has four increasing evidence levels:
\begin{table}[t]
\centering
\caption{A signature is not evidence that its transaction committed.}
\label{tab:receipt-levels}
\small
\begin{tabularx}{\columnwidth}{L{1.55cm} X}
\toprule
\textbf{Level} & \textbf{Established claim} \\
\midrule
Structural & The receipt has canonical syntax, typed bindings, and a valid
signature under its declared key. \\
Attested & A trusted kernel key attests the first terminal stage and reason;
this does not independently establish correct evaluation. \\
Replay & Retained inputs and pinned versions reproduce the stated decision or
transition. \\
Inclusion & A certified snapshot or log proof places the receipt and outcome,
and for Commit the head and lineage, in durable state. \\
\bottomrule
\end{tabularx}
\end{table}

The distinction matters because a kernel may sign a receipt before its storage
transaction aborts. Replay strength also depends on retained objects: deleting
an old proposal may preserve replay exclusion through its watermark while
removing the evidence needed to reproduce its original decision.
Appendix~\ref{app:receipt-semantics} specifies these verification obligations.

\subsection{Safety properties under the stated assumptions}

Under the conditional assumptions A1--A9 introduced in Section~\ref{sec:system-model} and detailed in Appendix~\ref{app:resource-artifact}, the protocol satisfies three primary safety properties:

\begin{proposition}[Owner-bound stable outcome]
\label{prop:stable-outcome}
Under A1--A4 and A8, only the principal named by an active allocation may
create $O[pid]$, and at most one disposition becomes durable for that
identifier.
\end{proposition}

\begin{proposition}[Single exact continuation]
\label{prop:single-successor}
Under A1--A4, at most one proposal commits from one complete predecessor or
serialized target absence, and the installed state is its sealed candidate.
\end{proposition}

\begin{proposition}[At-most-once accepted effect]
\label{prop:effect-once}
Under A1, A4, and A8, each effect identifier has at most one accepted binding,
including after online effect records are reclaimed.
\end{proposition}

The proofs are short serialization arguments and appear in
Appendix~\ref{app:lineage-proofs}. The propositions do not claim semantic
correctness: a policy may approve a false memory, and an idempotent intent may
still cause a non-idempotent remote action if its connector is faulty.

\subsection{Realization boundary}

The transaction contains no model call, network fetch, or human interaction;
those finish during preparation. A relational database may lock the branch and
named context rows. An object store may conditionally install one immutable
commit manifest, and a replicated service may serialize one activation command.
In every case readers reject a head whose accepted unit is incomplete.

\CK atomicity ends at its state boundary. External actions should be committed
as intents and delivered through an idempotent outbox when possible. No local
protocol can make an irreversible action atomic with an unrelated remote
system. Similarly, strict revocation requires the revocation row to share the
activation serialization order; a remote revocation service provides only a
declared bounded-staleness guarantee.

\section{Branch Lifecycle and Restoration}
\label{sec:branch-isolation-fencing}

Copying state, transferring write authority, and restoring old content are
different operations. Treating all three as ``load checkpoint'' can silently
fork a branch or revive an obsolete credential. \CK gives each operation a
typed forward transition.

\subsection{Branch creation}
\label{sec:branch-genesis}

A new branch has no predecessor. Its genesis record binds the authenticated
creation request, $d_\xi$, fresh identifier, initial roots/writer, and absence
proof. Candidate derivation uses declared genesis inputs; authorization uses
serialized $\xi_k$, never nonexistent predecessor state or authority. Both
freshness checks protect continued absence, $\xi_k$ versions/$dirSeq$, and all
bound genesis inputs. An optional source head is provenance, not a parent.

Branch creation runs Algorithm~\ref{alg:activate} with kind
\textsf{BranchCreate} and records the ordinary Commit disposition. One atomic unique insert installs the
initial state, admission context, genesis lineage, receipt, complete head, and
branch-directory row. Two concurrent
creators cannot both win the absent-key comparison. Every later transition has
an ordinary accepted parent. Reconciliation between branches is a later typed
proposal over exact source and target heads; the kernel does not choose a
semantic merge policy.

\subsection{Writer handoff}

A handoff uses a monotonically versioned directory record. Figure
\ref{fig:handoff-state-machine} shows the small state machine; Table
\ref{tab:handoff-actions} explains its actions. Every step compares the exact
current directory revision and appends a signed lifecycle result, yielding
maintenance receipts distinct from proposal dispositions.

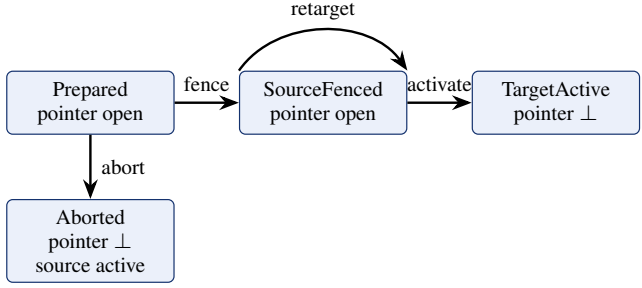
\begin{figure}[t]
\centering
\resizebox{0.99\linewidth}{!}{%
\begin{tikzpicture}[
  node distance=0.75cm,
  stage/.style={draw=navy, rounded corners=2pt, fill=sublayer,
    minimum width=1.95cm, minimum height=0.68cm, align=center},
  arr/.style={-{Stealth[length=2.3mm]}, thick},
  every node/.style={font=\scriptsize}
]
\node[stage] (p) {Prepared\\pointer open};
\node[stage, right=of p] (f) {SourceFenced\\pointer open};
\node[stage, right=of f] (a) {TargetActive\\pointer $\bot$};
\node[stage, below=0.75cm of p] (x) {Aborted\\pointer $\bot$\\source active};
\draw[arr] (p) -- node[above]{fence} (f);
\draw[arr] (f) -- node[above]{activate} (a);
\draw[arr] (f.north west) to[bend left=55] node[above]{retarget} (f.north east);
\draw[arr] (p) -- node[right]{abort} (x);
\end{tikzpicture}}
\caption{Handoff first removes the source writer, then activates the target.
Retargeting advances the epoch while the branch remains fenced.}
\label{fig:handoff-state-machine}
\end{figure}

\begin{table}[t]
\centering
\caption{Directory-serialized handoff actions.}
\label{tab:handoff-actions}
\small
\begin{tabularx}{\columnwidth}{L{1.15cm} X}
\toprule
\textbf{Action} & \textbf{Atomic directory effect} \\
\midrule
Prepare & Append Prepared and open its pointer; source may still advance. \\
Abort & Append Aborted, clear the pointer, and keep the source active. \\
Fence & Append SourceFenced, advance the open pointer, clear writer, advance epoch. \\
Retarget & Append SourceFenced and advance the open pointer, target, and epoch. \\
Activate & Run Algorithm~\ref{alg:activate} against the fenced revision and
append TargetActive, clear the pointer, and make the target writer active. \\
\bottomrule
\end{tabularx}
\end{table}

The gap between Fence and Activate is intentional. After fencing, no writer is
authorized; therefore a crash cannot expose two active writers. A retry either
uses the same handoff revision or fails after retargeting. Target activation is
the intentional exception to writer-bearing admission: it requires the exact
Frozen row, writer $\bot$, open SourceFenced revision, target, canonical epoch,
separately reserved successor epoch, and current admission context. Its
ordinary Commit receipt is indexed by both $O$ and $H_R$, so there is one
authoritative activation receipt.

Prepared recovery may retry, fence, or abort. SourceFenced recovery may retry
target activation or append a retarget, but it never returns to Prepared. If
directory and branch storage cannot share atomic ordering, a consensus or
coupled protocol is required; polling alone cannot establish writer isolation.

\subsection{Migration and forward restoration}

Migration names source and target schemas plus a pinned deterministic migrator.
The migrator must be allowed by current policy, terminate with a well-formed
target, and obey an explicit loss policy. A rollback is another forward
migration, never a rewrite of an accepted head.

Restoration also creates a new successor. Let $X_c$ be current state, $X_k$ an
authenticated historical checkpoint, $M_\mu$ an optional migration, and $M$ a
typed path mask. The restored candidate is
\begin{equation}
X'=\operatorname{Restore}_M(X_c,M_\mu(X_k)),\qquad
\Gamma'=\Gamma_c.
\label{eq:compact-restore}
\end{equation}
Paths in $M$ take checkpoint values after migration; all other paths take
current values. Authority, revocation, writer epoch, and lifecycle fields are
protected and therefore remain current. The successor's parent is the current
head, while the checkpoint is an auxiliary provenance reference. Restoration
cannot truncate lineage, revive an old writer or credential, or undo an
external action already performed. Appendix~\ref{app:lifecycle-semantics}
gives the total postconditions and stale-request behavior for every lifecycle
action.

\section{Evaluation}
\label{sec:evaluation}
\label{sec:executable-model}

We evaluate the Continuity Kernel (\CK) design by addressing four key research questions:
\begin{itemize}
  \item \textbf{RQ1 (Activation Integrity \& Succession):} Does the 17-stage activation predicate enforce exact predecessor succession, pre-state authorization, and non-self-authorizing state transitions across all five transition kinds?
  \item \textbf{RQ2 (Concurrency \& Replay Isolation):} Does the protocol maintain at-most-once execution identity, prevent replay attacks, and safely reclaim online proposal/effect records under concurrent races and watermark advances?
  \item \textbf{RQ3 (Lifecycle Isolation \& Writer Fencing):} Do lifecycle rules for branch creation, writer handoff, retargeting, aborts, and forward restoration prevent split-brain access and stale writer state transitions?
  \item \textbf{RQ4 (Self-Critical Scope \& Realization Limits):} What are the exact bounds of the formal state-space exploration, and what system-level guarantees require physical storage-engine verification beyond the abstract model?
\end{itemize}

\subsection{Formal State-Space Exploration (RQ1--RQ3)}

To evaluate the protocol's state-transition logic, we construct an executable bounded model in Python (\texttt{artifacts/bounded\_model.py}). The model performs exhaustive breadth-first search (BFS) state exploration over a finite abstraction of the kernel specification using only the Python standard library.

The model state space incorporates:
(1)~one active branch directory $B[k]$ and state store $X[k]$;
(2)~a shared admission context $\Gamma$ and creation context $\Xi[k]$;
(3)~an honest proposal owner and an adversarial principal;
(4)~proposal identifiers $pid \in \{0 \dots 12\}$ (13 IDs) and effect identifiers $eid \in \{0 \dots 3\}$ (4 IDs);
(5)~source writer $w_0$ and target writers $w_1, w_2$; and
(6)~integer versioning for schema, policy, and authority contexts.

Transitions simulate all five closed transition kinds (\textsf{Ordinary}, \textsf{Migration}, \textsf{Restoration}, \textsf{HandoffActivate}, \textsf{BranchCreate}), exact-head concurrency races, pre-state authority changes, missing/malformed evidence, two-stage freshness expiration, resubmissions, watermark reclamation, writer fencing, retargeting, and masked restoration.

\begin{table}[t]
\centering
\caption{State-space expansion and BFS exploration metrics by search depth $d$.}
\label{tab:depth-scaling}
\small
\begin{tabularx}{\columnwidth}{C{0.8cm} R{1.5cm} R{1.5cm} R{1.5cm} R{1.0cm}}
\toprule
\textbf{Depth} & \textbf{Unique States} & \textbf{Cum. States} & \textbf{Cum. Attempts} & \textbf{Time (s)} \\
\midrule
0 & 1 & 1 & 0 & <0.01 \\
1 & 15 & 16 & 17 & <0.01 \\
2 & 160 & 176 & 270 & 0.04 \\
3 & 1,455 & 1,631 & 2,959 & 0.33 \\
4 & 11,309 & 12,940 & 27,353 & 2.69 \\
5 & 76,028 & 88,968 & 216,413 & 18.88 \\
6 & 445,649 & 534,617 & 1,483,284 & 124.62 \\
7 & 2,273,613 & 2,808,230 & 8,880,248 & 27.45 \\
\bottomrule
\end{tabularx}
\end{table}

\begin{table}[t]
\centering
\caption{Distribution of transition outcomes across full state space ($8.88\text{M}$ transition attempts).}
\label{tab:disposition-breakdown}
\small
\begin{tabularx}{\columnwidth}{L{3.4cm} R{1.4cm} C{1.2cm}}
\toprule
\textbf{Transition Outcome} & \textbf{Full Count} & \textbf{Share (\%)} \\
\midrule
IdReuseConflict & 1,673,300 & 18.84\% \\
Commit (Head Advance) & 916,956 & 10.33\% \\
ReclaimedEffects & 842,195 & 9.48\% \\
RejectStaleHead & 770,555 & 8.68\% \\
PrefixNotFinal & 731,466 & 8.24\% \\
Defer (Missing Evidence) & 418,375 & 4.71\% \\
RejectDuplicateInProposal & 418,375 & 4.71\% \\
RejectUnauthorized & 395,905 & 4.46\% \\
RejectExpiredAtActivation & 395,905 & 4.46\% \\
RejectCandidateBinding & 395,905 & 4.46\% \\
RejectPolicy & 395,905 & 4.46\% \\
Prepared (Handoff) & 306,121 & 3.45\% \\
Other Dispositions ($11$ kinds) & 1,219,285 & 13.73\% \\
\midrule
\textbf{Total Attempts} & \textbf{8,880,248} & \textbf{100.00\%} \\
\bottomrule
\end{tabularx}
\end{table}

Under CPython~3.13.9, the depth-seven search evaluates 8,880,248 total transition attempts (including idempotency checks and retries). Of these, exactly
\begin{equation}
N_{\text{states}}=2,808,230, \qquad N_{\text{transitions}}=5,526,474
\end{equation}
are unique state-changing transitions where the successor state differs from the predecessor state.
As shown in Table~\ref{tab:depth-scaling}, state space exploration scales exponentially up to depth 7. Depth 6 timing (124.62s) includes evaluating all 7.4M outgoing transition attempts to generate depth 7 states, whereas depth 7 timing (27.45s) reflects terminal invariant validation on the $2.27\text{M}$ boundary states without further successor expansion. On a standard single-threaded execution harness, kernel transition evaluation achieves a throughput of \textbf{31,132 transitions/sec} with an average evaluation latency of \textbf{32.12~$\mu$s per transition}.

Table~\ref{tab:disposition-breakdown} details the empirical distribution of transition dispositions across all 8.88M evaluated transition attempts. Successful state transitions (\textsf{Commit}) represent 916,956 committed heads (10.33\%), while concurrency and identifier reuse protections (\textsf{RejectStaleHead}, \textsf{IdReuseConflict}, \textsf{ReclaimedEffects}) account for 37.00\% (3.28M transitions), proving that the protocol actively isolates concurrent races and stale proposals.

\begin{table}[t]
\centering
\caption{Reached coverage witnesses and invariant assertions ($2,808,230$ reachable states, depth 7).}
\label{tab:evaluation-coverage}
\small
\begin{tabularx}{\columnwidth}{L{2.3cm} X C{1.2cm}}
\toprule
\textbf{Target Category} & \textbf{Evaluated Invariant / Witness} & \textbf{Status} \\
\midrule
Dispositions & Reached Commit, Reject, Quarantine, and Defer & Pass \\
Activation Guard & 17 stages evaluated in normative sequential order & Pass \\
Authorization & Self-authorizing proposal ($\tau.a$) rejected at pre-state & Pass \\
Freshness & Revalidated initially and immediately prior to commit & Pass \\
Reclamation & Safe watermark advances for $pid$ and $eid$ scopes & Pass \\
Exact Succession & At-most-one accepted successor per complete head & Pass \\
Writer Fencing & Stale/abandoned writer blocked from branch mutation & Pass \\
Handoff Lifecycle & SourceFenced, TargetActive, Retarget, and Abort & Pass \\
Restoration & Masked restored root with protected current fields & Pass \\
\bottomrule
\end{tabularx}
\end{table}

As summarized in Table~\ref{tab:evaluation-coverage}, the exploration finds zero encoded invariant violations across all 2.8 million reachable states and reaches 100\% of the named protocol coverage witnesses.

\subsection{Invariant \& Safety Analysis}

The model validates three core structural properties across every reachable state:

\paragraph{1. Pre-State Authorization Safety (RQ1).}
For every committed transition, authorization is checked against the predecessor authority context $\Gamma$ (or genesis context $\xi_k$ for \textsf{BranchCreate}). In the model, when an adversarial proposer submits a self-authorizing proposal that attempts to inject its own required permissions into $a$, the \textsf{Authorization} stage evaluates $\operatorname{Authorize}_{\Gamma}(\tau, W)$ and rejects 395,905 invalid proposals (4.46\% of transitions) with \textsf{RejectUnauthorized}.

\paragraph{2. Execution Identity \& Replay Safety (RQ2).}
Each proposal identifier $pid$ has at most one durable disposition $O[pid]$, and each effect identifier $eid$ is bound at most once in $E[eid]$. Watermark advances correctly transition covered online identifiers to \textsf{RetiredIdentifier}, preventing re-execution or identifier recycling attacks after online metadata is pruned.

\paragraph{3. Writer Fencing \& Lifecycle Isolation (RQ3).}
During writer handoff, once the source writer is fenced (\textsf{SourceFenced}), any subsequent write attempt by the old writer $w_0$ fails with \textsf{StaleRevision} or \textsf{WriterEpoch}. Retargeting to a new target writer $w_2$ advances the writer epoch, ensuring that an abandoned target writer $w_1$ cannot activate the branch.

\subsection{Self-Critical Thesis Analysis \& Realization Limits (RQ4)}

While bounded model exploration provides strong evidence for the logical consistency of the finite protocol, a self-critical assessment highlights critical gaps between the formal model and physical system implementations:

\paragraph{1. Abstraction vs. Physical Storage Realization.}
The model treats preparation and atomic activation as single logical state steps. In a production system, atomic activation relies on the underlying storage engine (e.g., PostgreSQL conditional updates, FoundationDB OCC, or Raft consensus). Physical storage crashes during write-ahead logging (WAL), network partitions, or storage-engine bugs fall outside the model's abstract state space.

\paragraph{2. Signature vs. Durable Inclusion.}
As established in Section~\ref{sec:receipt-verification}, a signed candidate seal or preparation receipt proves cryptographic origin (Attested level) but not transactional durability (Inclusion level). If a kernel gateway signs a receipt in memory but the backing transaction fails to commit, the signed receipt is uncommitted. Production deployments must issue Level 4 inclusion proofs backed by durable log anchors.

\paragraph{3. External Side-Effect Atomicity.}
As discussed in Section~\ref{sec:transition-protocol}, \CK atomicity governs state transitions \emph{inside} its mutation boundary. It cannot make irreversible remote side effects (e.g., external API calls or physical robot actuation) atomic with internal state updates. Remote actions must be staged as idempotent outbox intents.

\paragraph{4. Performance \& Latency Trade-offs.}
The protocol adds overhead during commit: 17 sequential stage checks, double freshness validation, and prospective unit construction (32.12~$\mu$s in single-threaded Python). In high-throughput environments, this overhead requires storage-level optimizations, such as single-pass SQL transaction procedures or batched multi-key compare-and-swap operations.

\section{Related Work}
\label{sec:related-work}

\paragraph{Agent memory.}
Generative Agents and MemGPT organize reflection, retrieval, and long-term
context~\cite{park2023generative,packer2023memgpt}; LoCoMo, A-MEM, Mem0, and
MemOS study long-horizon evaluation, consolidation, scalable memory, and
versioned management~\cite{maharana2024locomo,xu2025amem,chhikara2025mem0,li2025memos}.
These systems motivate governed memory. As part of the broader Persistent Cognitive
Identity (\PCI) framework~\cite{pci2026stand}, \CK addresses the narrower question of
how any typed component version becomes the authoritative branch head.

MemTX is the closest agent-memory transaction system: it stages
evidence-bearing beliefs under a snapshot, gates actions, and propagates
repairs after retraction~\cite{li2026memtx}. \CK applies an activation boundary
to typed persistent state and specifies pre-state authority, complete-head
equality, stable four-way outcomes, writer epochs, and forward restoration.
MemTxn supplies source-supported memory admission, visible-version selection,
snapshots, and journal-based recovery~\cite{cui2026memtxn}. MemTX's unit is a
belief commit and MemTxn's is a source-supported memory update; \CK specifies a
typed branch-head transition and its pre-state authority boundary.

\paragraph{Transactions and retries.}
Optimistic concurrency control validates read assumptions at commit, database
transactions provide atomic durability, and replicated terms fence stale
leaders~\cite{kung1981optimistic,gray1992transaction,ongaro2014search}. For
linearizable retries, RIFL (Reusable Infrastructure for Linearizability)
combines unique request identifiers, atomic completion records, and safe
reclamation~\cite{lee2015rifl}. \CK composes
these mechanisms into an agent-state contract: the accepted unit also binds
typed state, authority, evidence, lineage, effects, and a calibrated receipt.
Commit-time authorization for LLM agents similarly exposes the gap between
temporary authority and durable effects~\cite{santosgrueiro2026committime}; \CK
places that check inside a persistent-state transition.

\paragraph{Effects and evidence.}
Cordon defines task-level semantic transactions with staged effects,
delegated authorization, compensation, and audit evidence~\cite{chen2026cordon}.
Its task boundary complements \CK's branch-head boundary: an outbox intent may
be a \CK component, but \CK cannot make an unrelated remote action atomic.
PROV-DM, JSON-LD, RDFC-1.0, and Data Integrity provide provenance,
canonicalization, and cryptographic proof formats
~\cite{w3c2013provdm,w3c2020jsonld,w3c2024rdfcanon,w3c2025dataintegrity}.
Those standards can show derivation and integrity; the activation protocol is
still needed to determine which validly encoded proposal became authoritative.

\section{Conclusion}
\label{sec:conclusion}

For persistent agents, state retention is not authority. Infrastructural continuity requires an authorized lineage of accepted branch heads. \CK makes the transition to authority explicit: untrusted components propose typed candidates off-commit; a deterministic control plane validates an exact predecessor head and pre-state authority (or typed absence and genesis context); a short activation transaction records one stable disposition; and only \textsf{Commit} advances the branch head. The protocol contract addresses linearizable retries, effect uniqueness, writer epoch fencing, schema migration, and forward restoration without expanding the kernel's trusted semantic scope.

Our safety propositions are conditional on explicit serialization, access-control, cryptographic, and durability assumptions. In a depth-seven finite abstraction, an executable bounded model explores 2,808,230 reachable states and 5,526,474 state-changing transitions, finding zero encoded invariant violations and reaching 100\% of named coverage witnesses. The model validates logical protocol consistency within its finite bounds; it does not establish unbounded mathematical correctness, physical storage driver conformance, or throughput performance under hardware faults. The core contribution of this work is the activation contract itself—providing a principled systems foundation that separates stored retention from authoritative reachability and distinguishes cryptographic receipt signatures from durable transactional inclusion.

\let\oldthebibliography\thebibliography
\let\endoldthebibliography\endthebibliography
\renewenvironment{thebibliography}[1]{%
  \begin{oldthebibliography}{#1}%
    \fontsize{6.0}{6.8}\selectfont
    \setlength{\itemsep}{0.0pt plus 0.1pt}%
    \setlength{\parskip}{0pt}%
}{%
  \end{oldthebibliography}%
}

\enlargethispage{3\baselineskip}
\bibliographystyle{unsrt}
\bibliography{references}

\appendix
\small
\section{Normative Receipt Semantics}
\label{app:receipt-semantics}

\subsection{Stage Sequences and Disposition Mapping}

This appendix fixes the stage order and minimum contents left implicit in the
main-text receipt abstractions. Preparation and activation checks evaluate two
strictly ordered, closed sequences:

\noindent\textbf{Preparation Stage Order (\textsf{PrepStage}):}
\[
\begin{aligned}
\textsf{RawScreen} &\prec \textsf{Decode} \prec \textsf{Canonicalize} \\
&\prec \textsf{ProposalAuth} \prec \textsf{StaticValidation} \\
&\prec \textsf{TargetLookup} \prec \textsf{AcquisitionAuth} \\
&\prec \textsf{RequirementCoverage} \prec \textsf{WitnessAuth} \\
&\prec \textsf{AcquisitionTime} \prec \textsf{CandidateDerivation} \\
&\prec \textsf{PolicyAuthorization} \prec \textsf{Ready}.
\end{aligned}
\]

\noindent\textbf{Activation Check Order (\textsf{ActCheck}):}
\[
\begin{aligned}
&\textsf{TargetLookup} \prec \textsf{Allocation} \prec \textsf{ExpectedHead} \\
&\quad \prec \textsf{Lifecycle} \prec \textsf{HandoffRevision} \\
&\quad \prec \textsf{AcquisitionAuth} \prec \textsf{RequirementCoverage} \\
&\quad \prec \textsf{VersionBinding} \prec \textsf{WitnessBinding} \\
&\quad \prec \textsf{CandidateBinding} \prec \textsf{EffectBinding} \\
&\quad \prec \textsf{FreshnessInitial} \prec \textsf{Authorization} \\
&\quad \prec \textsf{Decision} \prec \textsf{AuthorityTransition} \\
&\quad \prec \textsf{FreshnessFinal} \prec \textsf{WellFormedness}.
\end{aligned}
\]

In both pipelines, TargetLookup requires existence for existing-branch kinds
and absence plus serialized $\xi_k$ for BranchCreate. The latter's preparation
derives from genesis inputs and runs PolicyAuthorization against $\xi_k$;
activation ExpectedHead accepts only $\mathsf{Absent}$. Existing kinds use the
current predecessor $X,\Gamma$ and $\mathsf{Present}(d_h)$.

An outcome-eligible failure at position $j$ records
\begin{equation}
v_i=\begin{cases}
\mathsf{Pass},&i<j,\\
\mathsf{Fail}(r)\text{ or }\mathsf{Missing}(r),&i=j,\\
\mathsf{NE},&i>j,
\end{cases}
\label{eq:first-failure-vector}
\end{equation}
where \textsf{NE} means not evaluated. A verifier never asks a failing
predicate to pass. Known permanent failures precede a trusted Missing, so
unavailability cannot mask a rejection. Malformed outer framing, failed outer
authentication, and invalid allocation are protocol errors: they create
neither $O[pkey]$ nor a receipt.

\begin{table*}[t]
\centering
\caption{The four proposal-receipt variants and their terminal semantics.}
\label{tab:receipt-variants}
\scriptsize
\begin{tabularx}{\textwidth}{L{0.7cm} L{2.35cm} L{3.15cm} L{2.0cm} >{\raggedright\arraybackslash}X}
\toprule
\textbf{Tag} & \textbf{Terminal location} & \textbf{Stage/reason rule} &
\textbf{Builder} & \textbf{Durable effect} \\
\midrule
$r_P$ & Preparation & First failed or trusted-Missing preparation stage;
later stages are \textsf{NE} & $rb_P$ & Insert one owner-bound
$O[pkey]$ and $\mathcal R[d_{r_P}]=r_P$; no candidate becomes authoritative. \\
$r_A$ & Serialized activation & First failed activation check; earlier checks
pass and later checks are \textsf{NE} & $rb_A$ & Insert one
$O[pkey]$ and $\mathcal R[d_{r_A}]=r_A$; $X,B,\Gamma,E$ are unchanged. \\
$r_Q$ & $\mathsf{Decision}=\mathsf{Quarantine}$ & Earlier checks pass,
Decision records Quarantine, later checks are \textsf{NE} & $rb_Q$ & Insert
$O[pkey]$, $\mathcal R[d_{r_Q}]=r_Q$, and unreachable
$Q[pkey]$ only. \\
$r_C$ & All activation checks Pass & $P_\chi$ passes
\textsf{WellFormedness}; the terminal result is Commit & $rb_C$ & Store
$\mathcal R[d_{r_C}]=r_C$ and install the
complete accepted unit of Equation~\ref{eq:atomic-accepted-unit}. \\
\bottomrule
\end{tabularx}
\end{table*}

For $T\in\{P,A,Q,C\}$, $rb_T$ canonically encodes and signs exactly $r_T$,
computes $d_{r_T}$, and uses the sole receipt-store rule
$\mathcal R[d_{r_T}]=r_T$. No other proposal-receipt builder or receipt key
exists. Table~\ref{tab:closed-reasons} is the closed stage--reason map.

\begin{table*}[t]
\centering
\caption{Closed V1 stage--reason map. ``Limit'' abbreviates
\textsf{CanonicalizationLimit}; a stage vector disambiguates repeated reasons.}
\label{tab:closed-reasons}
\scriptsize
\begin{tabularx}{\textwidth}{>{\raggedright\arraybackslash}X >{\raggedright\arraybackslash}X}
\toprule
\textbf{Preparation stage $\mapsto$ allowed reason} &
\textbf{Activation check $\mapsto$ allowed reason} \\
\midrule
\textsf{RawScreen}$\mapsto\{$\textsf{Limit}$\}$;
\textsf{Decode}$\mapsto\{$\textsf{DecodeFailure, Limit}$\}$;
\textsf{Canonicalize}$\mapsto\{$\textsf{CanonicalizationFailure, Limit}$\}$;
\textsf{ProposalAuth}$\mapsto\{$\textsf{ProposalAuthentication}$\}$;
\textsf{StaticValidation}$\mapsto\{$\textsf{MalformedProposal, EffectManifest,
DuplicateEffectInProposal}$\}$;
\textsf{TargetLookup}$\mapsto\{$\textsf{UnknownTargetAtPreparation,
CreationContextMissing, BranchAlreadyExists}$\}$;
\textsf{AcquisitionAuth}$\mapsto\{$\textsf{UntrustedRequirementResult,
AcquisitionAuthentication}$\}$;
\textsf{RequirementCoverage}$\mapsto\{$\textsf{RequirementCoverage}$\}$;
\textsf{WitnessAuth}$\mapsto\{$\textsf{WitnessAuthentication}$\}$;
\textsf{AcquisitionTime}$\mapsto\{$\textsf{ExpiredAtPreparation,
EvidenceUnavailable}$\}$;
\textsf{CandidateDerivation}$\mapsto\{$\textsf{CandidateDerivation}$\}$;
\textsf{PolicyAuthorization}$\mapsto\{$\textsf{PreparationUnauthorized}$\}$.
&
\textsf{TargetLookup}$\mapsto\{$\textsf{TargetMissingAtActivation,
TargetAlreadyExistsAtActivation, CreationContextMissingAtActivation}$\}$;
\textsf{Allocation}$\mapsto\{$\textsf{AllocationRevoked}$\}$;
\textsf{ExpectedHead}$\mapsto\{$\textsf{StaleHead}$\}$;
\textsf{Lifecycle}$\mapsto\{$\textsf{OrdinaryLifecycle, HandoffLifecycle}$\}$;
\textsf{HandoffRevision}$\mapsto\{$\textsf{StaleHandoffRevision, StaleDirSeq,
TargetWriter, StaleEpoch}$\}$;
\textsf{AcquisitionAuth}$\mapsto\{$\textsf{AcquisitionAuthentication}$\}$;
\textsf{RequirementCoverage}$\mapsto\{$\textsf{RequirementCoverage}$\}$;
\textsf{VersionBinding}$\mapsto\{$\textsf{StaleVersion}$\}$;
\textsf{WitnessBinding}$\mapsto\{$\textsf{WitnessBinding}$\}$;
\textsf{CandidateBinding}$\mapsto\{$\textsf{CandidateBinding}$\}$;
\textsf{EffectBinding}$\mapsto\{$\textsf{EffectManifest, EffectScope,
RetiredEffect, DuplicateEffect}$\}$;
\textsf{FreshnessInitial, FreshnessFinal}$\mapsto\{$\textsf{ExpiredAtActivation,
StaleHead, StaleEpoch, StaleDirSeq, StaleAuthority, StaleVersion,
StaleCreationContext, TargetAppearedAtActivation}$\}$;
\textsf{Authorization}$\mapsto\{$\textsf{Unauthorized}$\}$;
\textsf{Decision}$\mapsto\{$\textsf{PolicyReject, PolicyQuarantine}$\}$;
\textsf{AuthorityTransition}$\mapsto\{$\textsf{InvalidAuthorityTransition}$\}$;
\textsf{WellFormedness}$\mapsto\{$\textsf{InvalidSuccessor}$\}$.
\\
\bottomrule
\end{tabularx}
\end{table*}

At activation \textsf{TargetLookup}, \textsf{BranchCreate} evaluates target presence before creation context: $B[k]\neq\mathsf{Absent}$ yields \textsf{TargetAlreadyExistsAtActivation}; $B[k]=\mathsf{Absent}\land\Xi[k]\text{ missing}$ yields \textsf{CreationContextMissingAtActivation}.
Freshness checks evaluate in deterministic intra-stage order. For existing branches:
(1)~$t>\text{deadline}$ yields \textsf{ExpiredAtActivation};
(2)~head mismatch yields \textsf{StaleHead};
(3)~epoch mismatch yields \textsf{StaleEpoch};
(4)~$dirSeq$ mismatch yields \textsf{StaleDirSeq};
(5)~authority/revocation/policy mismatch yields \textsf{StaleAuthority};
(6)~pinned version mismatch yields \textsf{StaleVersion}. For \textsf{BranchCreate}:
(1)~$t>\text{deadline}$ yields \textsf{ExpiredAtActivation};
(2)~post-lookup target appearance ($B[k]\neq\mathsf{Absent}$) yields \textsf{TargetAppearedAtActivation};
(3)~creation context ($\Xi[k]$ missing/altered) yields \textsf{StaleCreationContext};
(4)~genesis profile/allocator version mismatch yields \textsf{StaleVersion}.
Every reason selects Reject ($rb_A$), requires preparation and prefix $\Pi_{j-1}$, forbids later fields, and leaves $X,B,\Gamma,\Xi,E$ unchanged. Earlier stages dominate later stages.

At preparation, a Fail selects Reject/$rb_P$ and only
\textsf{EvidenceUnavailable} is Missing, selecting Defer/$rb_P$. At activation,
every reason selects Reject/$rb_A$ except \textsf{PolicyQuarantine}, which
selects Quarantine/$rb_Q$; all-pass through \textsf{WellFormedness} selects
Commit/$rb_C$. Thus Algorithm~\ref{alg:activate} has no other terminal path.
Adding a stage or reason requires a new receipt schema version. An exact retry
returns the indexed result; any changed attempt under the same $pid$ conflicts.

Lifecycle maintenance is typed separately. Its closed actions are
\textsf{Prepare}, \textsf{Fence}, \textsf{Retarget}, and \textsf{Abort}. Its
closed results are \textsf{Applied}, \textsf{StaleHead},
\textsf{StaleRevision}, \textsf{StaleDirSeq}, \textsf{WriterEpoch},
\textsf{TargetConflict}, \textsf{InvalidState}, and \textsf{Unauthorized}.
$rb_L=\mathsf{Build}(\mathsf{LifecycleReceiptV1})$ binds the action identifier,
pre/post directory tuples, optional input/output revision digests, result,
trusted time, and signature. It is stored in $\mathcal R$ and indexed by the
typed $H_R$ key defined in Appendix~\ref{app:lifecycle-semantics}; it never
creates $O[pkey]$ and is not a fifth proposal disposition.
Its digest is $d_{r_L}=d_{\mathsf{LifecycleReceiptV1}}(r_L)$ under
Equation~\ref{eq:typed-digest-appendix}. Replay recomputes the action and its
pre/post tuples; inclusion requires a certified state containing $\mathcal R$,
$H_R$, and, for Applied, the assigned $F$ and $B$ records.

Lifecycle field presence is action-indexed. Every $r_L$ binds its common
fields, authenticated request, requested pre-directory tuple, result, and exact
post tuple. For \textsf{Applied}, the input/output revision states are
Prepare $\bot\!\to$ Prepared, Fence Prepared$\to$SourceFenced, Retarget
SourceFenced$\to$SourceFenced, and Abort Prepared$\to$Aborted. For an allowed non-Applied
result, post equals pre, output is forbidden, and neither $F$ nor $B$ changes;
Prepare forbids input, while the other actions require the requested input.
Let
\begin{align*}
S=\{&\textsf{StaleRevision},\textsf{StaleDirSeq},\textsf{WriterEpoch},\\
&\textsf{InvalidState},\textsf{Unauthorized}\}.
\end{align*}
The exact failure sets are
$S\setminus\{\textsf{StaleRevision}\}\cup
\{\textsf{StaleHead},\textsf{TargetConflict}\}$ for Prepare,
$S\cup\{\textsf{StaleHead}\}$ for Fence,
$S\cup\{\textsf{TargetConflict}\}$ for Retarget, and $S$ for Abort. No other
action/result pair is schema-valid.

\subsection{Field presence and early input binding}

After owner authentication, the gateway computes
\begin{equation}
d_{raw}=d_{\mathsf{WireInputV1}}(rawBytes).
\end{equation}
An early RawScreen, Decode, or Canonicalize receipt binds $d_{raw}$ but omits
$d_\tau$; there is not yet a canonical proposal to hash. Let $j$ be the
terminal stage. A receipt requires the common fields, the terminal
Fail/Missing observation, and exactly the fields whose producing stages
successfully completed before $j$. It forbids fields first produced at $j$ or
later. Thus later \textsf{NE} stages impose no field obligation. The only
profile choice is raw retention: $d_{raw}$ is required when
\textsf{rawRetention=true} and forbidden otherwise. Absence is encoded by the
typed variant, never a zero digest.

\begin{table*}[t]
\centering
\caption{Stage-indexed field production. A field is required exactly after its
producer passes and forbidden beforehand.}
\label{tab:receipt-fields}
\scriptsize
\begin{tabularx}{\textwidth}{L{2.8cm} L{3.3cm} >{\raggedright\arraybackslash}X}
\toprule
\textbf{Field group} & \textbf{Producer that must Pass} & \textbf{Contents} \\
\midrule
Common & Outcome-eligible owner authentication & Type/version, subject/branch, principal,
allocation scope, $pid$, tag, disposition, stage vector, reason, issue time,
kernel key and signature \\
Raw input & Request capture under the signed profile & $d_{raw}$; profile makes
it required or forbidden, never optional \\
Canonical proposal & \textsf{Canonicalize} & $d_\tau$, typed
$\mathsf{Present}(d_h)$/$\mathsf{Absent}$ expectation, kind reference, and
unverified signature/credential bytes \\
Authenticated proposal & \textsf{ProposalAuth} & Verified signer, owner,
principal, allocation scope, and proposal-authentication result \\
Acquisition & \textsf{RequirementCoverage} (preparation) & Ordered requirements/results, witness and
dependency-version commitments \\
Candidate package & \textsf{CandidateDerivation} & Candidate root, pinned interpreter,
candidate-seal key/signature, effect manifest, and direct applicable $d_\xi$ or
open SourceFenced-revision digest \\
Activation prefix $\Pi_j$ & Corresponding $\mathsf{ActCheck}$ stage &
Existing row or absence plus $d_\xi$; allocation; exact Present/Absent
expectation; lifecycle; open revision;
acquisition signer; exact cover; versions; witness binding; candidate binding;
effect binding; initial freshness; authorization; Decision; authority-after;
final freshness---each introduced only by its same-named successful stage \\
Prospective unit & \textsf{BuildProspective} after \textsf{FreshnessFinal} &
Produced successor/core/lineage inputs, receipt-independent branch-row plan ($\widehat B_\chi$), receipt-independent assignment plan ($\widehat\Delta_\chi$), and deterministic parameters for receipt-dependent finalization \\
Finalized Commit & \textsf{WellFormedness} Pass & Commit receipt/digest,
complete head, outcome, $H_R$, effects, receipt store, and directory binding \\
\bottomrule
\end{tabularx}
\end{table*}

Consequently, a \textsf{CandidateDerivation} failure forbids candidate fields.
Every activation receipt requires the completed preparation package, but an
early TargetLookup or ExpectedHead failure contains only its successful
activation prefix and terminal observation---never later freshness,
authorization, Decision, authority-after, or successor fields. A Decision
Reject or Quarantine binds that terminal policy result; it does not manufacture
an AuthorityTransition prefix. A BranchCreate prefix binds absence and
$d_\xi$, never a predecessor head, state, or authority. An InvalidSuccessor
receipt may bind the produced prospective fields but forbids every finalized
Commit field.

Exact evidence coverage is a bijection between the signed ordered requirement
sequence and acquired results: no omission, duplicate, substitution,
injection, or forbidden reordering is accepted. Each result is either a signed
witness or an authenticated \textsf{Missing}. A candidate seal authenticates
these bindings; it is neither pre-state authorization nor proof of durable
commit.

\subsection{Verification levels}

Verification is monotone but not automatic:
\begin{enumerate}
  \item \textbf{Structural authenticity} checks the receipt variant,
  canonical syntax, typed digests, signature, key scope, and internal bindings.
  \item \textbf{Kernel-attested reason} additionally establishes that the
  authenticated kernel reported the declared first terminal stage. It does not
  independently establish that the predicate was evaluated correctly.
  \item \textbf{Independent replay} recomputes each predicate through the
  terminal stage from retained inputs and pinned versions. For Commit it also
  recomputes candidate derivation, the authority transition, manifest,
  successor core, and lineage. Replay is unavailable when any required object
  or interpreter is absent.
  \item \textbf{Durable inclusion} verifies a certified snapshot, commit
  manifest, or replicated-log proof containing $O[pkey]$ and the receipt; Commit
  also requires the installed head and lineage. An archive proof terminates at
  a retirement root in such a certified state. A signature or locator alone is
  not inclusion.
\end{enumerate}
A conforming verifier first completes level~1, reports level~2 only for a
trusted kernel key, attempts level~3 only with a complete replay package, and
reports level~4 only with an inclusion proof. Thus a signed receipt from an
aborted storage transaction cannot be upgraded to a committed fact.

\section{Typed Lineage and Cycle-Free Construction}
\label{app:lineage-proofs}

\subsection{Acyclic Hash Graph and Construction Order}

\begin{table*}[t]
\centering
\caption{Exact typed construction dependencies. Every input appears earlier.}
\label{tab:construction-dependencies}
\fontsize{6.3}{7.0}\selectfont
\renewcommand{\arraystretch}{0.92}
\begin{tabularx}{0.99\textwidth}{L{0.45cm} L{2.75cm} >{\raggedright\arraybackslash}X >{\raggedright\arraybackslash}X}
\toprule
& \textbf{Object and domain} & \textbf{Immediate inputs} &
\textbf{Downstream references} \\
\midrule
1 & $d_\xi$ / \textsf{CreationContextV1}; $\Xi[k]$ & Namespace, allocator and
policy versions, $dirSeq$, authority root, genesis profile & BranchCreate
proposal, seal, authorization, freshness \\
2 & $d_{f_j^{SourceFenced}}$ / existing \textsf{HandoffRevisionV1}; $F[hid,j]$
& Prior digest, SourceFenced state, branch/head, target, epochs, directory tuple,
context versions & HandoffActivate proposal, seal, admission \\
3 & $d_\tau$ / \textsf{TransitionV1} & $pkey,k$, Present/Absent expectation,
$\chi$, operations, authority changes, dependencies, requirements, validity,
kind reference, and direct applicable $d_\xi$ or $d_{f_j^{SourceFenced}}$ &
Acquisition and candidate derivation \\
4 & $d_W$ / \textsf{AcquisitionV1} & $d_\tau$, ordered requirements/results,
signer/key identifiers and pinned versions; detached authentication is verified & Candidate seal \\
5 & $d_{X'},d_{\Gamma'},d_{M_E}$ / \textsf{StateRootV1},
\textsf{AuthorityStateV1}, \textsf{EffectManifestV1} & Ordered typed successor
component digests; ordered authority entries; ordered $(opIndex,eid,opDigest)$,
derived from initial inputs with $d_\tau,d_W$ and the interpreter & Seal, head core \\
6 & $g,d_g$ / \textsf{BranchGenesisV1} & $d_\tau,d_\xi$, absence
proof, initial roots/writer/zero counters, optional source provenance & Genesis parent \\
7 & $d_{seal}$ / \textsf{CandidateSealV1} & $d_\tau,d_W$, interpreter,
component/effect roots, and direct applicable $d_\xi$ or open-revision digest & Receipt \\
8 & $h_c,d_{h_c}$ / \textsf{HeadCoreV1} & Exact Equation~\ref{eq:compact-head} fields:
root, typed parent, versions, counters, time, kind, and closed $auxRef$ & Output revision, lineage \\
9 & $d_{f_{j+1}^{TargetActive}}$ / new \textsf{HandoffRevisionV1} &
$d_{f_j^{SourceFenced}},d_{h_c}$, exact target, epochs, directory tuple and
context versions & Lineage, receipt, HResult \\
10 & $\ell,d_\ell$ / \textsf{LineageEdgeV1}; $L[\lambda]$ & Parent, $d_{h_c},d_\tau,pkey$,
versions, and applicable input/output handoff digests & Commit receipt \\
11 & $r_C,d_{r_C}$ / \textsf{CommitReceiptV1}; $\mathcal R[d_{r_C}]$ &
$pkey,d_\tau,d_W,d_{seal}$, parent, $d_{h_c},d_\ell$, all-Pass vector,
$\widehat P_\chi$ prospective fields, effect/authority and handoff bindings & Complete head, accepted records \\
12 & $h',d_{h'}$ / \textsf{CompleteHeadV1} & $h_c,d_\ell,d_{r_C}$ & Outcome, branch row, effect, handoff result, directory binding \\
13 & \textsf{OutcomeV1}; $O[pkey]$ & $pkey$, Commit tag, $d_{r_C},h'$ & Retry, inclusion \\
14 & \textsf{EffectRecordV1}; $E[eid]$ & $eid,pkey$, operation index/digest,
$d_{r_C},d_{h'}$ & Inclusion, duplicate exclusion \\
15 & \textsf{HandoffResultV1}; $H_R[K_H]$ & $pkey,hid$, input/output
revision digests, $d_{r_C},h'$ & Retry, inclusion \\
16 & \textsf{BranchRowV1}; $B[k]$ & $k,h'$, status, writer, epoch, $dirSeq$,
open pointer, immutable $createdFrom$ & Authoritative lookup \\
17 & \textsf{DirectoryBindingV1}; $G[K_G]$ &
$k,pkey,\chi$, pre/post $dirSeq$, allocator/context versions, $d_{r_C},d_{h'}$ &
Freshness, inclusion \\
18 & $i$ / \textsf{InclusionProofV1}, later & Certified state containing rows
11--17 & External verifier; never an input to $h'$ \\
\bottomrule
\end{tabularx}
\end{table*}

Here $K_H=\mathsf{HandoffResultKeyV1}(pkey,hid)$ and
$K_G=\mathsf{DirectoryKeyV1}(k,dirSeq')$; these are definitions, not aliases
for differently keyed records.

Every content reference uses a type- and version-separated digest
\begin{align}
d_T(y)=H\bigl(&\operatorname{LP}(\mathtt{CK})\parallel
\operatorname{LP}(T)\parallel\nonumber\\
&\operatorname{u32}(v)\parallel
\operatorname{LP}(\operatorname{Canon}_T(y))\bigr),
\label{eq:typed-digest-appendix}
\end{align}
where \textsf{LP} is an unambiguous length prefix. Logical proposal and effect
identifiers are allocated names, not content digests.

Equation~\ref{eq:compact-head} is the sole definition of immutable head core,
complete head, and branch row; this appendix introduces no shortened head.
The closed parent type is
\begin{align}
\mathsf{LineageParent}={}&\mathsf{AcceptedParent}(d_h)\nonumber\\
&\mid\mathsf{GenesisParent}(d_g).
\end{align}
Exactly one branch-genesis edge uses the second variant. Every later edge uses
the complete current head. A checkpoint, migration input, or handoff source is
an auxiliary typed reference and can never occupy the parent field.

A \textsf{LineageEdgeV1} binds $sid,bid$, transition kind, lineage parent,
$d_{h_c}$, $d_\tau$, the proposal allocation key, policy and authority
versions, epoch and sequence, plus the transition-kind-specific source,
checkpoint, migration, or handoff references. A \textsf{BranchGenesisV1}
binds the authenticated creation request, absent-directory proof, initial state
and admission roots, writer, epoch~0, sequence~0, and optional source
provenance. Source provenance does not create cross-branch ancestry. In
Table~\ref{tab:construction-dependencies}, every \textsf{V1} name has version~1,
that exact name as its domain tag, and its type-specific canonical encoder under
Equation~\ref{eq:typed-digest-appendix}. Component roots carry their declared
type/schema domain. Rows 3--10 and 12 are content-addressed by their displayed
typed digest; rows 1--2, 11, and 13--17 use exactly the displayed map key. Map
keys are allocated typed names, not hidden hash inputs.

For HandoffActivate, \textsf{TransitionV1} directly binds
$d_{f_j^{SourceFenced}}$; \textsf{AcquisitionV1} binds it transitively through
$d_\tau$, and the seal binds it directly. The revision must be the current open
SourceFenced revision or \textsf{HandoffRevision} fails stale. $h_c.auxRef$
also names this input but never the output or seal. The head core has no direct
seal dependency: $h'$ binds $d_{seal}$ transitively through $d_{r_C}$. The
same rule applies to BranchCreate with $d_\xi$: direct in the transition and
seal, transitive through $d_\tau$ in the acquisition, and revalidated from
$\Xi[k]$ at both freshness checks. The
output revision binds the open input and $d_{h_c}$; successful activation then
clears the pointer while retaining both revisions in $F$. The order is
\begin{align}
(d_\xi?,d_{f_j^{SourceFenced}}?)&\prec d_\tau\prec d_W
\prec(X',\Gamma',M_E)\nonumber\\
&\prec d_g?\prec d_{seal}\prec h_c
\prec d_{f_{j+1}^{TargetActive}}?\nonumber\\
&\prec\ell\prec r_C\prec h'\prec d_{h'}\nonumber\\
&\prec(O,E,H_R?,B',G')\prec\mathsf{Publish}\prec i.
\label{eq:construction-order}
\end{align}
The receipt binds the prospective unit $\widehat P_\chi$, not the complete head; the
complete head $h'$ then binds $d_{r_C}$, producing $d_{h'}$. The lineage edge likewise binds $h_c$
and typed predecessor, not the complete successor. Therefore every digest edge
points left in Equation~\ref{eq:construction-order}, and the audited graph is
acyclic. Outcome, effect, handoff-result, branch, and directory records point
only to named earlier digests; no earlier object points back. Inclusion
evidence is constructed only after publication. Atomic publication changes
visibility, not this hash order.

\subsection{Proof sketches}

\begin{proof}[Proposition~\ref{prop:stable-outcome}]
Every outcome-writing path first authenticates the proposal allocation and
revalidates it in the proposal-scope serialization domain. Invalid outer
requests never write an outcome. The first authorized write uniquely inserts
$O[pkey]$; a compatible retry returns it, an incompatible retry conflicts, and
a reclaimed identifier is rejected by the monotone watermark. A second durable
disposition is therefore impossible under A1--A4 and A8.
\end{proof}

\begin{proof}[Proposition~\ref{prop:single-successor}]
The seal binds the proposal, evidence, interpreter, typed Present/Absent
expectation, candidate root, and effect manifest. Activation serializes the
branch key and compares either the complete head or continued absence plus
$d_\xi$. The first Commit changes that serialization point atomically; every
contender then fails exact-head or absence admission. Thus at most one sealed
candidate becomes the exact successor or genesis under A1--A4.
\end{proof}

\begin{proof}[Proposition~\ref{prop:effect-once}]
The operation--manifest bijection rejects an effect identifier repeated within
a proposal. Across proposals, Commit uniquely inserts $E[eid]$ while holding
its scope. Reclamation advances the scope watermark before deleting online
rows, and scopes are never recycled. Hence deletion cannot reopen an accepted
identifier under A1, A4, and A8.
\end{proof}

These are conditional serialization arguments. They establish neither content
truth nor correctness of policy, canonicalization, cryptography, or storage
implementations.

\section{Total Lifecycle, Handoff, and Restoration Semantics}
\label{app:lifecycle-semantics}

\subsection{Handoff Union Types and Kind-Specific Admission}

Equation~\ref{eq:compact-head} is normative for $B$. In particular, $dirSeq$
lives only in $B$; a revision records $dirSeqIn,dirSeqOut$ as snapshots.
$createdFrom\in\{\mathsf{None},\mathsf{SourceProvenance}(d_h)\}$ is assigned
at genesis and copied byte-for-byte thereafter. Pointer activity and revision
state are distinct closed types:
\begin{equation}
\begin{aligned}
\mathsf{HPtr}&=\bot\mid\mathsf{Open}(hid,d_f),\\
\mathsf{FState}&=\mathsf{Prepared}\mid\mathsf{SourceFenced}\mid{}\\[-1mm]
&\qquad\mathsf{Aborted}\mid\mathsf{TargetActive}.
\end{aligned}
\label{eq:handoff-union}
\end{equation}

\begin{table*}[t]
\centering
\caption{Transition-kind-specific head admission within the single ordered
$\mathsf{ActCheck}$ vector.}
\label{tab:admission-by-kind}
\fontsize{6.3}{7.0}\selectfont
\renewcommand{\arraystretch}{0.92}
\begin{tabularx}{\textwidth}{L{2.1cm} >{\raggedright\arraybackslash}X >{\raggedright\arraybackslash}X}
\toprule
\textbf{Kind} & \textbf{Serialized $G^{kind}_{head}$ precondition} &
\textbf{Parent and disposition} \\
\midrule
Ordinary, Migration, Restoration &
$B=\langle h,\mathsf{Active},w,e,d,z,c\rangle$; exact $expected=h$, writer
$w$, epoch $e$, and $dirSeq=d$; an Open $z$ must reference Prepared, never
SourceFenced. &
$\mathsf{AcceptedParent}(d_h)$; one of the four proposal dispositions. \\
HandoffActivate &
$B=\langle h_s,\mathsf{Frozen},\bot,e_c,d,
\mathsf{Open}(hid,d_{f_j^{SourceFenced}}),c\rangle$; exact fenced head, target writer,
$dirSeq=d$, input revision and current context; the revision separately binds
$reservedEpoch=e_r$, and admission requires the proposal to bind both
$e_c,e_r$ with $e_c=e_r$. &
$\mathsf{AcceptedParent}(d_{h_s})$; an ordinary Commit/Reject/Quarantine/Defer
outcome under $O[pkey]$. \\
BranchCreate &
$B[k]$ absent, $expected=\mathsf{Absent}$, and exact serialized $\Xi[k]=\xi_k$;
authenticated allocation/absence proof and zero initial counters. &
$\mathsf{GenesisParent}(d_g)$; ordinary Commit with kind
\textsf{BranchCreate}. \\
\bottomrule
\end{tabularx}
\end{table*}

An ordinary Commit during Prepared is nonconflicting but makes that revision's
source-head binding stale. Fence then returns \textsf{StaleHead}; Abort followed
by a fresh Prepare is required before fencing. Any change to target, epoch,
$dirSeq$, input revision, $d_G$, or the current admission versions similarly
invalidates an older target package.

\subsection{Authoritative State Postconditions}

Let $b=\langle h,s,w,e,d,z,c\rangle$ follow the field order of
Equation~\ref{eq:compact-head}; $b[a\leftarrow x]$ copies $b$ and explicitly
replaces field $a$. The HResult map is $H_R$. Let
$K_L=\mathsf{LifecycleKeyV1}(hid,lid)$ and
$K_H=\mathsf{HandoffResultKeyV1}(pkey,hid)$ be disjoint typed keys;
$lid$ and $pid$ are owner-allocated, non-recycled names. A unique insert
makes each $H_R$ key immutable: an exact retry returns the stored value and a
variant conflicts. ``Append $r_L$'' below means
$\mathcal R[d_{r_L}]=r_L$ in the same atomic unit.
For compactness, let
$\mathcal S=\langle X,\Gamma,rev,issuer,policy,\Xi,h,L,O,E,G\rangle$;
$\mathcal S'=\mathcal S$ explicitly copies every listed family, including $\Xi'=\Xi$.

\begin{table*}[t]
\centering
\caption{Total lifecycle postconditions. Every authoritative family is
assigned explicitly or through the complete accepted unit; $\Xi'=\Xi$ holds for all rows.}
\label{tab:lifecycle-total}
\fontsize{6.3}{7.0}\selectfont
\renewcommand{\arraystretch}{0.92}
\begin{tabularx}{\textwidth}{L{1.2cm} L{2.85cm} >{\raggedright\arraybackslash}X >{\raggedright\arraybackslash}X}
\toprule
\textbf{Action} & \textbf{Exact precondition} &
\textbf{Complete authoritative poststate} & \textbf{Receipt/result records} \\
\midrule
Genesis & $B[k]$ absent; BranchCreate row of Table~\ref{tab:admission-by-kind}
& Install $\mathcal U_{accept}(\mathsf{BranchCreate})$: $X_0,\Gamma_0,g,h_0,L,
E[M_E],G'_\chi$ and
$B'[k]=\langle h_0,\mathsf{Active},w_t,0,0,\bot,c_0\rangle$; $\Xi'=\Xi$;
$F,H_R$ not written. & $O[pkey]=\mathsf{OutcomeV1}(\mathsf{Commit},d_{r_C},h_0)$ and
$\mathcal R[d_{r_C}]=r_C$. \\
Prepare & $b.status=\mathsf{Active}$, $b.writer=w_s$, exact head/epoch/$dirSeq$,
and $b.handoff=\bot$ & $\mathcal S'=\mathcal S$;
$B'=b[dirSeq\leftarrow d+1,handoff\leftarrow
\mathsf{Open}(hid,d_{f_{j+1}^{Prepared}})]$, preserving $createdFrom$; append
$f_{j+1}^{Prepared}$ to $F$ and reserve $e_r=e+1$. & Append
$r_L=rb_L(\mathsf{Applied})$ and
$H_R[K_L]=\mathsf{Lifecycle}(d_{r_L},\bot,d_{f_{j+1}^{Prepared}})$. \\
Abort & Exact $f_j^{Prepared}$; Active, writer $w_s$, epoch $e$, $dirSeq=d$,
and matching Open pointer & $\mathcal S'=\mathcal S$;
$B'=b[dirSeq\leftarrow d+1,handoff\leftarrow\bot]$, preserving
$createdFrom$; append the Aborted revision to $F$ and cancel $e_r$. & Append $r_L$ and
$H_R[K_L]=\mathsf{Lifecycle}(d_{r_L},d_{f_j^{Prepared}},d_{f_{j+1}^{Aborted}})$. \\
Fence & Exact $f_j^{Prepared}$, its final $h$, and active source tuple &
$\mathcal S'=\mathcal S$; preserve $createdFrom$ and set
$B'.status=\mathsf{Frozen}$, $writer=\bot$, $epoch=e_r$, $dirSeq=d+1$, and
$handoff=\mathsf{Open}(hid,d_{f_{j+1}^{SourceFenced}})$; append SourceFenced revision
with old epoch $e$ and new/reserved epoch $e_r$. & Append $r_L$ and
$H_R[K_L]=\mathsf{Lifecycle}(d_{r_L},d_{f_j^{Prepared}},d_{f_{j+1}^{SourceFenced}})$. \\
Retarget & Exact $f_j^{SourceFenced}$; Frozen, writer $\bot$, epoch $e$, $dirSeq=d$,
and matching Open pointer & $\mathcal S'=\mathcal S$;
preserve status, writer, head, and $createdFrom$; set epoch $e+1$, $dirSeq=d+1$, and
$handoff=\mathsf{Open}(hid,d_{f_{j+1}^{SourceFenced}})$; append SourceFenced revision
naming the new target and reserved epoch $e+1$. & Append $r_L$ and
$H_R[K_L]=\mathsf{Lifecycle}(d_{r_L},d_{f_j^{SourceFenced}},d_{f_{j+1}^{SourceFenced}})$. \\
Activate & HandoffActivate row of Table~\ref{tab:admission-by-kind} &
Install $\mathcal U_{accept}(\mathsf{HandoffActivate})$: $X'_T,h',L,E[M_E],
G'_\chi$, append $f_{j+1}^{TargetActive}$, set $\Xi'=\Xi$ and
$B'=\langle h',\mathsf{Active},w_t,e_c,d+1,\bot,c\rangle$; install the
current-or-authorized $\Gamma_T,rev_T,issuer_T,policy_T$. &
$\mathcal R[d_{r_C}]=r_C$, $O[pkey]=\mathsf{OutcomeV1}(\mathsf{Commit},d_{r_C},h')$, and
$H_R[K_H]=\mathsf{Activation}(d_{f_j^{SourceFenced}},\allowbreak
d_{f_{j+1}^{TargetActive}},\allowbreak d_{r_C},\allowbreak h')$: the receipt
digest is identical. \\
Restore & Ordinary-kind row; authenticated checkpoint/proof/mask and optional
pinned migrator & Install $\mathcal U_{accept}(\mathsf{Restoration})$ with
$X'=\operatorname{Restore}_M(X_c,M_\mu(X_k))$, new $h',L,E[M_E],G'_\chi$, $\Xi'=\Xi$ and
$B'=b[head\leftarrow h']$; status, writer, epoch, $dirSeq$, handoff,
$createdFrom$, and protected authority families stay current. & Ordinary
Commit in $O$ and $\mathcal R$; no $H_R$ or $F$ write; checkpoint is auxiliary
lineage, never a parent. \\
Failure & Any authenticated maintenance action whose exact precondition fails
& $\mathcal S'=\mathcal S$, $B'=B$, and $F'=F$; no accepted state,
effect, directory, or lifecycle-revision write. &
Append $r_L=rb_L(q_L)$ and the action-indexed $H_R[K_L]$ from Appendix~A;
the output revision is forbidden. \\
\bottomrule
\end{tabularx}
\end{table*}

Every \textsf{HandoffRevisionV1} binds its predecessor revision, action/state,
branch key, source head/writer, target writer, old/new and reserved epochs,
$dirSeqIn,dirSeqOut$, pinned policy/authority/revocation versions, digests of
$\Gamma,G$, optional $d_{h_c}$, reason, and trusted time. $F[hid,j]$ is
append-only. Maintenance changes only the fields and indexes assigned in
Table~\ref{tab:lifecycle-total}; it never creates an $O$ entry.

The consumed target package is
\begin{align}
P_H=\langle&pkey,d_\tau,d_W,d_{seal},hid,d_{f_j^{SourceFenced}},d_{h_s},\nonumber\\
&w_t,e_c,e_r,dirSeq,policyV,authV,revV,d_\Gamma,d_G\rangle,
\label{eq:target-bindings}
\end{align}
where $e_c$ is the pre-activation canonical epoch and $e_r$ is the separately
reserved successor epoch. Admission requires both to remain exact and equal.
The accepted unit creates the distinct output
$f_{j+1}^{TargetActive}$; Appendix~\ref{app:lineage-proofs} fixes every digest
placement. Precisely, $\Gamma_T=\Gamma_{current}$ when no separately
authorized authority transition exists, and $\Gamma_T=\Gamma'$ only when the
same accepted proposal authorizes and binds $\Gamma'$. The identical rule
applies to revocation, issuer, and policy; the fenced historical head is never
their source.

\subsection{Restoration Path Mask and Protected Semantics}

Restoration binds
\begin{align}
aux_R=\langle&d_{h_c},d_{root_c},d_{h_k},d_{root_k},\nonumber\\
&proof_k,d_\mu,M,componentRoots\rangle
\end{align}
\noindent\begin{minipage}{\columnwidth}
and computes $X'=\operatorname{Restore}_M(X_c,M_\mu(X_k))$. For protected paths
\begin{align}
P=\{&authority,policy,revocation,issuer,\nonumber\\
&writer,epoch,dirSeq,lifecycle,createdFrom\},\nonumber\\
M\cap P&=\varnothing,\qquad X'|_P=X_c|_P.
\label{eq:restore-protected}
\end{align}
Malformed or unauthenticated outer lifecycle requests create no receipt.
\end{minipage}

\section{Deterministic Resource Profile and Artifact}
\label{app:resource-artifact}
\enlargethispage{2\baselineskip}

\begin{table*}[t]
\centering
\caption{Exact counter increments. An occurrence is counted even if later
deduplicated; tests happen before the event that would exceed its bound.}
\label{tab:counter-semantics}
\fontsize{6.3}{7.0}\selectfont
\begin{tabularx}{\textwidth}{L{2.45cm} >{\raggedright\arraybackslash}X L{2.6cm} >{\raggedright\arraybackslash}X}
\toprule
\textbf{Counter} & \textbf{Increment or depth event} &
\textbf{Counter} & \textbf{Increment or depth event} \\
\midrule
\textsf{inputOctets} & Before consuming each primary or supplied-context octet;
duplicate packaged objects count again. &
\textsf{decodeDepth} & Entering a JSON array/object; root depth is 1 and exit
decrements it. \\
\textsf{decodedItems} & Decoder emits each member, array element, or scalar
occurrence in the primary or each supplied context object. &
\textsf{expandedTerms} & JSON-LD expansion emits each node, property, or value
occurrence. \\
\textsf{contextUrls} & Before each JSON-LD Context Processing URL attempt after
base resolution; repeated URLs, repeated references, and cache hits count. &
\textsf{contextDepth} & Entering each logical nested Context Processing call;
root is 1, cache use changes neither entry nor exit. \\
\textsf{expansionDepth} & Entering a recursive JSON-LD expansion call; root is
1 and exit decrements it. &
\textsf{quads} & RDF conversion emits each quad occurrence, before set
deduplication. \\
\textsf{blankNodes} & First allocation or capture of each distinct blank-node
identifier. &
\textsf{nDegreeCalls} & Entry to each RDFC Hash N-Degree Quads invocation. \\
\textsf{permutationBranches} & Before exploring each RDFC permutation
candidate. &
\textsf{issuerEntries} & Before a logical issuer state first becomes reachable,
charge every mapping in that state: a fresh insertion charges one and a clone
charges its full inherited mapping cardinality. Equal mappings in distinct
issuers or clones count again; structural sharing gives no discount. \\
\textsf{pathUnits} & Each identifier, position, or digest token appended to an
N-degree path. &
\textsf{intermediateUnits} & Simultaneously with each increment of
decodedItems, expandedTerms, quads, permutationBranches, issuerEntries, or
pathUnits; hence it is their running sum. \\
\bottomrule
\end{tabularx}
\end{table*}

\subsection{Canonicalization resource profile}

Canonicalization is governed by the signed, versioned profile
\begin{align}
\kappa=\langle&jsonldV,rdfcV,contextDigests,\nonumber\\
&maxContextUrls,maxContextDepth,\nonumber\\
&maxBytes,maxDecodeDepth,\nonumber\\
&maxDecodedItems,maxQuads,\nonumber\\
&maxExpandedTerms,maxExpansionDepth,\nonumber\\
&maxBlankNodes,maxNDegreeCalls,\nonumber\\
&maxPermutationBranches,maxIssuerEntries,\nonumber\\
&maxPathUnits,maxIntermediateUnits\rangle.
\label{eq:resource-profile}
\end{align}
$jsonldV$ pins the 16 July 2020 JSON-LD~1.1 Processing Algorithms and API
Recommendation: \textsf{processingMode=json-ld-1.1}, its Context Processing
and Expansion algorithms, \textsf{ordered=true}, \textsf{rdfDirection=null},
a signed absolute base IRI, and \textsf{produceGeneralizedRdf=false}. Input is UTF-8
\textsf{application/ld+json} containing one JSON object or array; HTML
extraction and ambient bases are forbidden~\cite{w3c2020jsonld,w3c2020jsonldapi}.
Remote contexts are digest-pinned package objects, never network-fetched. Each
reference resolves against the signed base to one exact URL--digest manifest
entry; redirects, negotiation, and unlisted URLs are rejected. Primary and all
supplied context bytes are charged at RawScreen and decoded at Decode. Duplicate
supplied objects are charged separately; repeated references and cache hits do
not recharge bytes but do count logical context-processing attempts. $rdfcV$
pins RDFC-1.0~\cite{w3c2024rdfcanon}. The path is JSON decode, Context
Processing/Expansion, RDF emission, then RDFC canonicalization. Table
\ref{tab:counter-semantics} defines mathematical nonnegative counters, not
implementation allocations.
All cumulative counters reset once per proposal and sum occurrences across its
datasets; they never reset per recursive invocation. Depth counters are active
recursion gauges. \textsf{inputOctets} belongs to RawScreen,
\textsf{decodeDepth}/\textsf{decodedItems} to Decode, and the remaining counters
to Canonicalize.

Every bound is tested before its event; context URL/depth exhaustion terminates
at \textsf{Canonicalize} before lookup or recursive entry, with
\textsf{CanonicalizationLimit} and no partial candidate. For a clone, the full inherited charge is tested before the clone becomes
reachable; for insertion, the unit charge is tested before insertion. Either
failure terminates at \textsf{Canonicalize} with the same
\textsf{CanonicalizationLimit} receipt and no partial issuer state.
Crossing any other limit terminates at its governing \textsf{RawScreen},
\textsf{Decode}, or \textsf{Canonicalize} stage with
\textsf{CanonicalizationLimit}, produces no partial candidate, and binds
$d_{raw}$ rather than a nonexistent $d_\tau$. RDFC-1.0 separately identifies
adversarially expensive graph structure as a denial-of-service risk
~\cite{w3c2024rdfcanon}; this is \emph{canonicalization complexity poisoning},
not a semantic claim that the resulting graph is false. A local wall-clock,
memory, or process watchdog may stop work earlier, but that event is a transient
operational failure with no durable disposition unless the implementation can
translate it into the deterministic counters above.

\subsection{Conditional assumptions}

All safety claims are conditional on:
\begin{enumerate}\fontsize{8.5}{9.5}\selectfont\itemsep=-1pt\parskip=0pt
  \item[\textbf{A1}] \label{assumption:a1}\textbf{Mediation.} No bypass
  credential or API can advance an authoritative head or protocol index.
  \item[\textbf{A2}] \label{assumption:a2}\textbf{Encoding and cryptography.}
  Canonical encodings, typed hash domains, and signature schemes have their
  stated security and interoperability properties.
  \item[\textbf{A3}] \label{assumption:a3}\textbf{Key custody.} Proposal,
  allocator, evaluator, preparation-service, candidate-seal, and kernel keys
  are protected and restricted to their declared principals and scopes.
  \item[\textbf{A4}] \label{assumption:a4}\textbf{Atomic serialization.} The
  complete accepted unit, lifecycle changes, and reclamation changes serialize
  over every named mutable key and are durably all-or-none.
  \item[\textbf{A5}] \label{assumption:a5}\textbf{Complete context.} Policy,
  evaluator, authority, revocation, dependency, allocation, and lifecycle
  versions are locally lockable or transactionally revalidated.
  \item[\textbf{A6}] \label{assumption:a6}\textbf{Trusted time.} Activation
  obtains a conservative commit-time interval or an equivalent storage-enforced
  deadline.
  \item[\textbf{A7}] \label{assumption:a7}\textbf{Lifecycle order.} Branch
  creation, epochs, handoff, and recovery share one order observed by writers.
  \item[\textbf{A8}] \label{assumption:a8}\textbf{Non-recycled scopes.}
  Authenticated allocators issue monotone proposal/effect identifiers;
  retirement watermarks advance before covered rows are deleted.
  \item[\textbf{A9}] \label{assumption:a9}\textbf{Verification objects.}
  Replay and inclusion are claimed only when their required objects,
  interpreters, keys, and certified-state proofs are available.
\end{enumerate}
Availability is not assumed. Receipt structural authenticity and kernel
attestation use A1--A3; independent replay additionally uses the relevant
A5--A6 objects and A9; durable inclusion additionally uses A4, A8, and A9.

\subsection{Executable Model and Artifact Specification}

The formal state-space exploration of Section~\ref{sec:evaluation} is provided as an executable Python artifact (\texttt{artifacts/bounded\_model.py}). The artifact imports only standard library modules (\texttt{collections}, \texttt{dataclasses}, \texttt{typing}) and executes an exhaustive breadth-first search (BFS) state exploration over the finite protocol state space.

The executable model verifies that across all 2,808,230 reachable states and 5,526,474 state-changing transitions:
(1)~no encoded invariant is violated;
(2)~every named protocol coverage witness is reached; and
(3)~all terminal dispositions, 17-stage check sequences, pre-state authority rules, writer fencing, and forward restoration postconditions hold.
Reproduction instructions and SHA-256 integrity digests are documented in \texttt{artifacts/README.md}.

\flushcolsend
\end{document}